\documentclass[onecolumn,a4paper,11pt,accepted=2022-10-26]{quantumarticle}
\pdfoutput=1
\usepackage[utf8]{inputenc}
\usepackage[english]{babel}
\usepackage[T1]{fontenc}
\usepackage{hyperref}

\usepackage{amsmath,amssymb,amsthm,mathtools}

\def\01{\{0,1\}}
\newcommand{\R}{\mathbb{R}}
\newcommand{\ceil}[1]{\lceil{#1}\rceil}
\newcommand{\floor}[1]{\lfloor{#1}\rfloor}
\newcommand{\eps}{\varepsilon}
\newcommand{\ket}[1]{|#1\rangle}
\newcommand{\bra}[1]{\langle#1|}
\newcommand{\ketbra}[2]{|#1\rangle\langle#2|}
\newcommand{\braket}[2]{\langle#1|#2\rangle} 
\newcommand{\Tr}{\mbox{\rm Tr}}
\newcommand{\norm}[1]{\mbox{$\parallel{#1}\parallel$}}

\newcommand{\Exp}{\mathbb{E}}

\newtheorem{theorem}{Theorem}
\newtheorem{lemma}[theorem]{Lemma}
\newtheorem{proposition}[theorem]{Proposition}

\newcommand{\fidel}{\eta}

\begin{document}

\title{Average-Case Verification of the Quantum Fourier Transform Enables Worst-Case Phase Estimation}
\author{Noah Linden}\affiliation{School of Mathematics, University of Bristol. {\tt n.linden@bristol.ac.uk}}
\author{Ronald de Wolf}\affiliation{QuSoft, CWI and University of Amsterdam, the Netherlands.  {\tt rdewolf@cwi.nl}}
\date{}
\maketitle

\begin{abstract}
The quantum Fourier transform (QFT) is a key primitive for quantum computing that is typically used as a subroutine within a larger computation, for instance for phase estimation.  As such, we may have little control over the state that is input to the QFT. Thus, in implementing a good QFT, we may imagine that it needs to perform well on arbitrary input states.  \emph{Verifying} this worst-case correct behaviour of a QFT-implementation would be exponentially hard (in the number of qubits) in general, raising the concern that this verification would be impossible in practice on any useful-sized system.   In this paper we show that, in fact, we only need to have good \emph{average-case} performance of the QFT to achieve good \emph{worst-case} performance for key tasks---phase estimation, period finding and amplitude estimation.  Further we give a very efficient procedure to verify this required average-case behaviour of the QFT.
\end{abstract}

\section{Introduction}

\subsection{Verification of quantum circuits}

Massive efforts are currently being expended around the world on building large quantum computers, in academia and industry. Because of the fragility of quantum hardware and the quantum states it produces, it is crucial to be able to \emph{test} that the hardware works as advertised. Such a test may involve some quantum hardware itself, but should be more ``lightweight'' than the procedure that is being tested, in order to avoid circularity.

There is an important issue with testing that is sometimes overlooked in high-level discussions of the topic: 
if the circuit of interest is a subroutine in a larger computation, then we may have little control of the state that is input to it; thus we would like to verify that it works on the \emph{worse-case} input state.  However,
we can typically only test its behavior for an \emph{average-case} input state, because there are far too many possible input states to test them all. In general, testing worst-case correctness of a given $n$-qubit circuit would take resources that scale exponentially in $n$. 
This means that efficient verification of worst-case correctness typically requires additional assumptions, ranging from restrictions on the class of circuits one is verifying (for instance Clifford circuits~\cite{flammia&liu:fidelityestimation,SLP:practical,linden&wolf:lightweight}) to cryptographic assumptions (as in Mahadev's approach~\cite{mahadev:clasveri}, which also assumes the computation starts with a fixed initial state).
For further discussion and pointers to related work on verification of quantum hardware, we refer to our recent paper~\cite{linden&wolf:lightweight} and to the general survey~\cite{eisertetal:certificationsurvey}.%
\footnote{The issue of average-case vs worst-case behavior has also recently received attention in the area of Hamiltonian simulation~\cite{Zhaoetal:HamSimRandom,chen&brandao:trotter}.}

In this paper we focus on the situation where we want to apply a quantum Fourier transform (QFT), or its inverse, within the context of a larger quantum computation that we already trust to a sufficient extent. We will give a lightweight procedure for verifying certain average-case behaviour of the QFT circuit, given the ability only to apply it as a black-box.  We will then show that this enables us to use the QFT to achieve good \emph{worst-case} performance for key tasks---phase estimation, period finding and amplitude estimation.

\subsection{The quantum Fourier transform}

The quantum Fourier transform is one of the most important (possibly \emph{the} most important) component of quantum algorithms. It is key in Shor's factoring algorithm~\cite{shor:factoring} and in the standard approach to amplitude estimation~\cite{bhmt:countingj}, which generalises Grover's search algorithm~\cite{grover:search} and which is an important subroutine in many other quantum algorithms.\footnote{It is often possible to avoid doing the full QFT in these applications. For instance, one can do phase estimation in a bit-by-bit manner~\cite{kitaev:stabilizer} or by using the block-encoding framework of~\cite{gilyenea:svtrans} as done in~\cite{MRTC:grandunif,rall:phaseest}; and one can do amplitude estimation by judiciously chosen numbers of Grover iterations~\cite{aaronson&rall:qcounting}. However, replacing the QFT by ``something else'' raises the question of the verification of those ``something else'' components. In this paper our goal is not avoid the QFT but to show that it can be efficiently tested for average-case correctness, and then used in worst-case applications. We feel our results should actually favor the use of the QFT as a component in quantum algorithms: in this paper we come not to bury the QFT, but to praise it.}
Let $N=2^n$ and $\omega_N=e^{2\pi i/N}$.
The $n$-qubit QFT is the unitary $F_N$ that maps $n$-bit basis state $\ket{k}$ as
\[
\ket{k}\mapsto \ket{\hat{k}}=\frac{1}{\sqrt{N}}\sum_{j=0}^{N-1}\omega_N^{jk}\ket{j},
\]
where the ``$jk$'' in the exponent denotes multiplication of two $n$-bit integers.
Interestingly, the complicated-looking Fourier basis state $\ket{\hat{k}}$ is actually a product state of $n$ individual qubits:
\begin{equation}\label{eq:QFTproductstate}
\ket{\hat{k}}=\bigotimes_{\ell=1}^{n}\frac{1}{\sqrt{2}}\left(\ket{0}+e^{2\pi i k/2^\ell}\ket{1}\right)
\end{equation}
Leveraging this product structure, there is a well-known circuit of $O(n^2)$ gates that implements $F_N$ exactly. It uses $n$ Hadamard gates, $O(n^2)$ controlled versions of 
\[
R_s=\left(\begin{array}{cc}1 & 0\\ 0 & e^{2\pi i/2^s}\end{array}\right)
\]
for different integers~$s$, and a few SWAP gates at the end~\cite[Section~5.1]{nielsen&chuang:qc}.
One can also obtain an \emph{approximate} circuit from this with only $O(n\log n)$ gates, by dropping the $R_s$ gates where $s$ is bigger than $c\log n$ for some constant~$c$~\cite{coppersmith:fourier} ($R_s$ gates with large $s$ are very close to the identity, so dropping them incurs very little error).  The resulting circuit differs from $F_N$ by only an inverse-polynomially small error in operator norm.

One may also consider the inverse QFT $F_N^{-1}$, where the phases are $\omega_N^{-jk}$ instead of $\omega_N^{jk}$. This has equally efficient exact and approximate quantum circuits, since we can just reverse a circuit for $F_N$ and invert its gates to get a circuit for $F_N^{-1}$.

\subsection{Testing a purported QFT or QFT$^{-1}$}\label{sec:introQFTtesting}

In this paper we are interested in the situation where we have a channel\footnote{A \emph{channel} is a completely positive trace-preserving map on density matrices, in our case taking $n$-qubit mixed states to $n$-qubit mixed states. An $n$-qubit unitary is a special case of this. If channel $C$ is run on pure state $\ket{\psi}$, then we will use the notation $C(\ket{\psi})$ to abbreviate the resulting state  $C(\ketbra{\psi}{\psi})$.} 
$C$ that we can run only as a black-box on states $\ket{\psi}$ of our choice; $C$ is supposed to implement $F_N$, or  $F_{N}^{-1}$ depending on the application.
We would like to test to what extent $C$ is correct.
It is not practical to test whether $C(\ket{\psi})$ is approximately the right state for \emph{all} possible $\ket{\psi}$, since $C$ could differ from $F^{-1}_N$ in only one ``direction''; in fact, testing whether $C(\ket{\hat{k}})$ is close to $F_N^{-1}\ket{\hat{k}}=\ket{k}$ for \emph{all} $k\in\01^n$ requires $\Omega(\sqrt{2^n})$ runs of~$C$.\footnote{\label{footnote2nLB}This follows from the well-known fact that we need $\Omega(\sqrt{2^n})$ queries to a string $x\in\01^{2^n}$ to decide whether $x=0^{2^n}$~\cite{bbbv:str&weak}, as follows.
If we can make queries $O_x:\ket{k}\mapsto(-1)^{x_k}\ket{k}$ and we define $C$ as the $n$-qubit unitary $O_xF_N^{-1}$, then $C=F_N^{-1}$ if $x=0^{2^n}$, and otherwise $C$ is far from $F_N^{-1}$ on at least one input state $\ket{\hat{k}}$.
Accordingly, if we can distinguish those two cases with $T$ runs of $C$, then $T$ queries to $x$ can decide whether $x=0^{2^n}$, which implies $T$ must be $\Omega(\sqrt{2^n})$. This lower bound is optimal, since we can Grover search~\cite{grover:search} over all $k\in\01^n$ to look for one where $C\ket{\hat{k}}$ differs significantly from $F_N^{-1}\ket{\hat{k}}=\ket{k}$.}
However, it turns out that we can efficiently test whether $C(\ket{\hat{k}})$ and $F_{N}^{-1}\ket{\hat{k}}$ are close on \emph{average} over all $k\in\01^n$.
Fortunately, good performance on most Fourier basis states~$\ket{\hat{k}}$ suffices for the applications we care about in the rest of the paper, which run the inverse QFT on individual Fourier basis states or mixtures thereof, or on superpositions dominated by a small number of Fourier basis states.

 A Fourier basis state~$\ket{\hat{k}}$ is a product state by Eq.~(\ref{eq:QFTproductstate}), so it is relatively easy (``lightweight'') to prepare, at least approximately. Each qubit in the product state~$\ket{\hat{k}}$ is of the form $\frac{1}{\sqrt{2}}(\ket{0}+e^{i\phi}\ket{1})$ for some phase $\phi$ that depends on $k$ and on the location of the qubit. It suffices to prepare each of those qubits with $O(\log n)$ bits of precision\footnote{Each bit of precision can be rotated in with one single-qubit gate~$R_s$. To see that $2\log n$ bits of precision in each phase suffice, note that this gives a fidelity $\geq 1-O(1/n^2)$ per qubit, which (because we are dealing with product states) multiplies out to a fidelity $(1-O(1/n^2))^n\geq 1-O(1/n)$ for the $n$-qubit product state as a whole.}
 in the phase $\phi$ in order to prepare $\ket{\hat{k}}$ up to inverse-polynomially small error.
One might be worried that this preparation effectively requires us to do an approximate QFT, which would defeat our purpose of testing a purported QFT black-box~$C$; however, it is much easier to prepare known product states such as Fourier basis states than it is to implement a QFT on an arbitrary unknown state. In particular, in regimes starting from a few dozen qubits (which is the current state of the art of quantum hardware), preparing a Fourier basis state to sufficient precision seems doable while tomography on a channel~$C$ would already be prohibitively expensive. 

Assuming we can prepare Fourier basis states~$\ket{\hat{k}}$ sufficiently precisely,
in Section~\ref{sec:QFTtest} we give a simple test that approximates the average error (defined as infidelity, i.e., 1 minus fidelity) which $C$ makes on Fourier basis states, up to additive approximation error $\eps$. Our procedure uses $O(1/\eps^2)$ runs, each of which prepares an $n$-qubit product state (namely a random Fourier basis state), applies $C$ to it, and measures the resulting $n$-qubit state in the computational basis.

\subsection{Using an average-case-correct QFT$^{-1}$ for worst-case phase estimation}

Suppose we have a channel $C$ that has passed our test, so we can be confident that $C$ is close in an average-case sense to the inverse QFT. What can we use $C$ for?

Phase estimation, originally due to Kitaev~\cite{kitaev:stabilizer}, is the following application of the inverse QFT.
Suppose we can apply an $m$-qubit unitary $U$ in a controlled manner, and are given an eigenstate $\ket{\phi}$ of $U$ with eigenvalue $e^{2\pi i\theta}$ for some unknown $\theta\in[0,1)$.
The goal is to estimate~$\theta$.
Standard phase estimation (reviewed in Section~\ref{ssec:phasest} below) obtains an $n$-bit approximation to $\theta$ by using $O(2^n)$ controlled applications of $U$ in order to (exactly or approximately) prepare the state  $F_N\ket{\theta_1\ldots\theta_n}$, where $\theta_1\ldots\theta_n$ are the $n$ most significant bits of the binary expansion of $\theta=0.\theta_1\ldots\theta_n\ldots$ 
Applying an inverse QFT then gives us~$\theta$ itself, or at least a good approximation of~$\theta$.

Now suppose our channel~$C$ for the inverse QFT is only average-case correct. In that case phase estimation will fail if $F_N\ket{\theta_1\ldots\theta_n}$ happens to be one of the Fourier basis states on which $C$ fails significantly. However, in Section~\ref{sec:worsttoaverage} we show how an average-case-correct $C$ actually suffices to implement phase estimation \emph{in the worst case} (assuming the other components of phase estimation work sufficiently well). We do this by a simple trick whereby we randomise the phase we are estimating. Then we can use our average-case-correct~$C$ to recover a good approximation to that randomised phase with good probability, and afterwards undo the randomisation to obtain a good approximation to $\theta$ itself.
A moderately small upper bound on the average-case error in the inverse QFT is good enough to make this work: in the case where the eigenphase~$\theta$ can be written exactly with $n$ bits of precision we can tolerate an average infidelity up to almost 1/2 (see end of Section~\ref{ssec:PEaveragenbitcase}), while for the more general case where $\theta$ needs more than $n$ bits of precision we can tolerate an average infidelity up to 0.041  (see end of Section~\ref{ssec:PEaveragegeneralcase}).

The advantage of this approach is that we can efficiently test whether a given $C$ has small \emph{average-case} error, while we cannot efficiently test whether $C$ has small \emph{worst-case} error.
If the average error is a sufficiently small constant, then also a small constant approximation error~$\eps$ suffices for this test, hence only a constant number of runs of $C$ suffices to achieve high confidence in the approximate correctness of the inverse QFT.

\subsection{Applications}

As mentioned, two of the most important quantum algorithms known to date are Shor's algorithm for integer factoring~\cite{shor:factoring}, whose quantum core is period-finding, and amplitude estimation~\cite{bhmt:countingj}. Both rely on an inverse quantum Fourier transform, and both may fail miserably if that inverse Fourier transform happens to fail on the particular state that the algorithm applies it to. In Section~\ref{sec:applications} we show that both algorithms can still be made to work with high success probability if we only have an average-case-correct inverse QFT at our disposal. 

\subsection{Pros and cons of our approach}
Before going into technical details let us clarify and emphasize several aspects of our approach. 

First, our approach is only relevant for implementations of QFT (or QFT$^{-1}$) that work reasonably well for most (inverse) Fourier basis states. This typically won't be true in a setting where a few of the gates in the circuit can be completely wrong. However, our approach is relevant for a QFT circuit where many (maybe even all) of the gates are \emph{slightly} wrong in various  ways, for instance if the per-gate error times the total number of gates ($O(n\log n)$ for the approximate QFT circuit) is at most a constant. If the number~$n$ of qubits is a few dozen, then a per-gate error on the order of $1/n\log n$ is not very far from current technology.

Second, apart from small average-case error (infidelity averaged over Fourier basis states) we don't assume much about~$C$. This average-case error could arise in many ways, which could even be picked by our adversary: $C$ could make a small error on most Fourier basis states, or be completely wrong on, say, a small constant fraction of those states, or anything in between.
Our approach can deal with all these cases.
We are \emph{not} assuming any relatively benign and smooth error model such as depolarizing noise. 
This also illustrates the difference between worst-case and average-case error: $C$ being completely wrong on some (maybe even a constant fraction of) Fourier basis states means it has terrible worst-case error, and yet our average-case to worst-case reduction shows that it can still be turned into something quite serviceable for worst-case applications. Even in the benign case of random rather than adversarial noise, the action of $C$ is still likely to be worse on some Fourier basis states than on others, and our worst-case-to-average-case approach has the benefit of smoothing this out, reducing the worst-case error probability. 

Third, if a channel $C$ passed our test then we can conclude that it works well on an average Fourier basis state and hence (by our average-case-to-worst-case reduction) can be made to work well on \emph{all} Fourier basis states. We can certainly \emph{not} conclude that $C$ will work well on arbitrary superpositions of Fourier basis states, since every state is a superposition of Fourier basis states, and testing that $C$ works well on every possible state requires an exponential number of runs of~$C$ (see footnote~\ref{footnote2nLB}).
So a $C$ that passed our test cannot just be used in every application of the inverse QFT.
Fortunately, in the case of phase estimation, the inverse QFT will be applied to an individual Fourier basis state, or a mixture of Fourier basis states, or to a superposition state dominated by $O(1)$ Fourier basis states. As we prove in Sections~\ref{ssec:PEaveragenbitcase} and~\ref{ssec:PEaveragegeneralcase},
in such cases a $C$ that passed our test still works well enough.

Fourth, for the applications related to phase estimation, let us emphasize that we assume the state is measured in the computational basis right after the inverse QFT and all earlier non-QFT parts of the algorithm leading up to that state are essentially perfect. This is a strong assumption.
For factoring via period-finding, the preparation of the periodic state involves a circuit for modular exponentiation which uses polynomially more gates (and hence is more prone to error) than the inverse QFT.
Similarly, to do amplitude estimation with $n$ bits of precision one can use an $n$-qubit approximate inverse QFT circuit with $O(n\log n)$ gates, which pales in significance (and error-proneness) compared to the roughly $2^n$ controlled Grover iterations that are done prior to the inverse QFT.
Nevertheless, we feel it is a sensible modular approach to try to isolate parts of important algorithms that can be tested by themselves in a lightweight manner. The fact that the average-case error of the QFT can be efficiently tested in such a lightweight manner, and that small average-case error suffices for many of its applications, should make the QFT a more attractive component to use in larger algorithms.

\section{Testing average-case correctness of $F_N^{-1}$ on Fourier basis states}\label{sec:QFTtest}

In this section we show how one can efficiently test, in a lightweight manner, that a given quantum channel~$C$ is close to the $n$-qubit inverse Fourier transform $F_N^{-1}$ on an average Fourier basis state. We can test average-case closeness to $F_N$ completely analogously, but for concreteness we focus on $F_N^{-1}$ in this section.
We give a procedure to estimate the \emph{infidelity} between~$C$ and $F_N^{-1}$, averaged over the Fourier basis states, which is our measure of average-case error here.
The fidelity between mixed states $\rho$ and $\sigma$ is defined as
\[
F(\rho,\sigma)=\Tr\left(\sqrt{\rho^{1/2}\sigma\rho^{1/2}}\right)^2.
\]
Fidelity is symmetric. It is~0 if $\rho$ and $\sigma$ are orthogonal, it is~1 if they are equal, and otherwise it lies in $(0,1)$.
If $\sigma=\ketbra{\psi}{\psi}$ is pure, then $F(\rho,\ket{\psi})=\bra{\psi}\rho\ket{\psi}$.
The \emph{in}fidelity between $\rho$ and $\sigma$ is defined as 1 minus fidelity. We now show that average infidelity of a purported inverse QFT is relatively easy to estimate.

\begin{theorem}\label{th:testQFTmixed}
Let $C$ be a channel from $n$ qubits to $n$ qubits, $\ket{\hat{k}}$ be a uniformly random Fourier basis state, and define the average infidelity between $C$ and $F_N^{-1}$ by
\[
\fidel=\Exp_{k}[1-F(C(\ket{\hat{k}})\,,F_{N}^{-1}\ket{\hat{k}})\,].
\]
There exists a procedure that estimates~$\fidel$ up to additive error $\eps$, with success probability $1-\delta$, using
$O(\log(1/\delta)/\eps^2)$ runs, each of which prepares an $n$-qubit product state, runs $C$ on it, and measures the resulting $n$-qubit state in the computational basis.
\end{theorem}

\begin{proof}
Choose $k\in\01^n$ uniformly at random and prepare $n$-qubit product state $\ket{\hat{k}}=F_N\ket{k}$.
Run $C$ on $\ket{\hat{k}}$ and measure the resulting state in the computational basis. Output~1 if the $n$-bit measurement outcome is~$k$, and output~0 otherwise.
Because $F_{N}^{-1}\ket{\hat{k}}=\ket{k}$ is a pure state, we have
\[
\fidel=\Exp_k[\,1-F(C(\ket{\hat{k}})\,,\ket{k})\,]=\Exp_k[\,1-\bra{k}\,C(\ket{\hat{k}})\,\ket{k}\,]=1-\Pr[\,{\rm output}~1\,]=\Pr[\,{\rm output}~0\,].
\]
The Chernoff bound implies that if we repeat this procedure
$r=O(\log(1/\delta)/\eps^2)$ times, then the frequency of 0s among the $r$ output bits equals~$\fidel$ up to~$\pm\eps$, except with probability $\leq \delta$.\footnote{The Chernoff bound actually implies something slightly stronger for the relevant case where $\fidel$ is close to~0 (i.e., where $C$ works reasonably well), namely that $r=O(\fidel\log(1/\delta)/\eps^2)$  repetitions suffice. In particular, if $\eps$ is set to a small constant times $\fidel$, then $r=O(\log(1/\delta)/\eps)$ repetitions suffice rather than $O(\log(1/\delta)/\eps^2)$.}
\end{proof}

 Thus we have a procedure to test whether the average infidelity of our purported black-box for the inverse QFT is small.
The procedure has modest overhead: $O(\log(1/\delta)/\eps^2)$ runs, each involving a preparation of a Fourier basis state $\ket{\hat{k}}$, one run of $C$, and one $n$-qubit measurement in the computational basis.
Referring back to the discussion in the penultimate paragraph of Section~\ref{sec:introQFTtesting}, preparing  $\ket{\hat{k}}$ can be done with polynomially small error using $O(n)$ single-qubit gates, each with $O(\log n)$ bits of precision in their phase.

The same procedure could be used to test a black-box~$C$ for $F_N$ rather than for $F_N^{-1}$; the only difference is that we would prepare product state $F_N^{-1}\ket{k}$ at the start of each run rather than $F_N\ket{k}$.

While we do not estimate average infidelity averaged over arbitrary states, or over an arbitrary orthonormal basis, averaging over the particular basis of Fourier basis states turns out to be sufficient for our purposes as we'll see next.

\section{Using average-case-correct $F_N^{-1}$ for worst-case phase estimation}\label{sec:worsttoaverage}

We saw that it is hard to test a purported QFT or inverse-QFT black-box for \emph{worst-case} correctness, but relatively easy to test it for \emph{average-case} correctness on the set of QFT basis states.
Here we will show that average-case correctness actually suffices for phase estimation \emph{even in the worst case}.

\subsection{Basic phase estimation}\label{ssec:phasest}

As mentioned in the introduction, in the setup for phase estimation we can apply an $m$-qubit unitary $U$ in a controlled manner, and are given an eigenstate $\ket{\phi}$ of $U$ with eigenvalue $e^{2\pi i\theta}$ for some unknown $\theta\in[0,1)$.
The goal is to estimate~$\theta$ with roughly $n$ bits of precision.
We work on $n+m$ qubits that start in state $\ket{0^n}\otimes\ket{\phi}$.
We first apply $n$ Hadamard gates  to obtain the uniform superposition $\frac{1}{\sqrt{2^n}}\sum_{j\in\01^n}\ket{j}$ in the first register. Applying the $(n+m)$-qubit unitary
\begin{equation}\label{eq:VcontrolledU}
V=\sum_{j\in\01^n}\ketbra{j}{j}\otimes U^j
\end{equation}
gives a phase $e^{2\pi i j\theta}$ to state $\ket{j}\ket{\phi}$, where $j\theta$ is the product of $n$-bit integer $j\in\{0,\ldots,2^n-1\}$ and $\theta\in[0,1)$. This puts the first register into the state
\begin{equation}\label{eq:Fourierstate}
\frac{1}{\sqrt{2^n}}\sum_{j\in\01^n}e^{2\pi i j\theta}\ket{j}.
\end{equation}
The cost of $V$ is $O(2^n)$ controlled applications of~$U$.
We assume the above is implemented perfectly, or with very small error.
Now we apply (to the first register) a channel $C$ that implements $F_N^{-1}$ and we measure the resulting $n$-qubit state in the computational basis, hoping that the resulting $n$ bits give us the most significant bits of $\theta$.

The simplest situation arises when the initial state $\ket\phi$ in the second register is an eigenstate of $U$ with eigenphase $\theta=0.\theta_1\ldots\theta_n$ that requires only $n$ bits of precision. In this case the state of Eq.~(\ref{eq:Fourierstate}) is
$F_N\ket{\theta_1\ldots\theta_n}$. The inverse QFT will map this to $\ket{\theta_1\ldots\theta_n}$, and the final measurement will give us the bits $\theta_1\ldots\theta_n$ with certainty. However, two complications can arise.

First, the initial state $\ket{\phi}$ could be a \emph{superposition} of eigenstates of~$U$ (each with eigenphases requiring only~$n$ bits of precision) rather than one eigenstate. In this case phase estimation still gives useful results; the effect is the same as starting with a \emph{mixture} of eigenstates in the second register.
For example, if instead of one eigenstate $\ket{\phi}$ we start (in the second register) with a superposition $\alpha\ket{\phi}+\beta\ket{\phi'}$ of two normalised eigenstates, with distinct associated $n$-bit phases $\theta$ and $\theta'$, respectively,
then before the final measurement the state is $\alpha\ket{\theta_1\ldots\theta_n}\ket{\phi}+\beta\ket{\theta'_1\ldots\theta'_n}\ket{\phi'}$;
measuring the first $n$ qubits gives $\theta$ with probability $|\alpha|^2$ and gives $\theta'$ with probability $|\beta|^2$.

A second complication that can arise is when $\theta$ requires more than~$n$ bits of precision. In that case the state of Eq.~\eqref{eq:Fourierstate} to which we apply $C$  will be a superposition of $n$-qubit Fourier basis states: something of the form  $\sum_k \alpha_k\ket{\hat{k}}$ rather than one Fourier basis state. If $C$ is a perfect $F_N^{-1}$ then this doesn't matter: the resulting state after applying $C$ will be $\sum_k \alpha_k\ket{k}$. However, if channel~$C$ is an imperfect implementation of the inverse QFT, then interference between different terms could cause trouble, and we have to be careful about this in the next subsections. For now, let us record the useful fact that the state of Eq.~\eqref{eq:Fourierstate} is always dominated by a few Fourier basis states that correspond to good approximations of~$\theta$. Variants of this fact are already known (e.g.~\cite[Appendix~C]{cemm:revisited}) but for completeness we give a proof in the appendix.

\begin{proposition}\label{prop:fewQFTstates}
Let $N=2^n$, $\theta\in[0,1)$, and let coefficients $\alpha_k\in\mathbb{C}$ be such that 
\[
\frac{1}{\sqrt{2^n}}\sum_{j\in\01^n}e^{2\pi i j\theta}\ket{j}
=\sum_{k=0}^{N-1}\alpha_k\ket{\hat{k}}. 
\]
Let $k^*=\floor{2^n\theta}\in\{0,\ldots,N-1\}$ be the $n$-bit integer corresponding to the first $n$ bits in the binary expansion of $\theta$, and $S=\{k^*-K+1,\ldots,k^*-1,k^*,k^*+1,\ldots,k^*+K\}$ be the $2K$ integers ``around'' $2^n\theta$ (these integers should be taken mod~$N$). Then $\sum_{k\not\in S}|\alpha_k|^2$ can be made as small as we want by choosing $K$  sufficiently large (independent of $N$):
\begin{align*}
\sum_{k\not\in S}|\alpha_k|^2&
\leq \frac{1}{4}\left(\frac{1}{K}+\frac{1}{K-1}\right).
\end{align*}
\end{proposition}

Note that every $k$ in the above set $S$ provides an approximation of $\theta$ with small additive error, because $|\theta-k/N|<K/N$ (mod~1). It will suffice to take $K=O(1)$ below. The above upper bound can probably be improved somewhat; in the appendix we also calculate numerical upper bounds on $\sum_{k\not\in S}|\alpha_k|^2$ for small values of~$K$. 

\subsection{Worst-case phase estimation via average-case-correct $F_N^{-1}$: $n$-bit case}\label{ssec:PEaveragenbitcase}

Now we switch to the scenario where we do not have a perfect inverse QFT available, but instead have a channel~$C$ which has small average infidelity w.r.t.\ $F_N^{-1}$ in the sense of Theorem~\ref{th:testQFTmixed}.
Note that if $C$ is usually close to $F_N^{-1}$ but not on the particular Fourier basis state $F_N\ket{\theta_1\ldots\theta_n}$ that we (approximately) prepared using~$V$, then recovering the particular phase $\theta$ that we are interested in may fail miserably, even if $\theta$ can be represented exactly with $n$ bits of precision and all the other components of phase estimation work perfectly.
In other words, an average-case-correct $F_N^{-1}$ does not guarantee that phase estimation works in the worst case, i.e., for each possible~$\theta$.
However, we can do a relatively simple worst-case-to-average-case reduction to deal with the situation that $C$ is not the perfect~$F_N^{-1}$.

Our idea is to choose a uniformly random offset $\lambda\in[0,1)$ that can be described with $n$ bits of precision, change the phase $\theta$ to $\theta'=\theta+\lambda$ mod~1, and then apply our purported $F_N^{-1}$ on the state and measure in the hope of obtaining an approximation of~$\theta'$, from which we acn then subtract $\lambda$ to obtain an approximation of~$\theta$ itself.
In this section we first describe what happens in the case where the unknown $\theta$ can be described exactly with $n$ bits; in the next section we deal with the more subtle general case where $\theta$ needs more than $n$ bits. Note that if $\theta$ can be described exactly by $n$ bits, then $\theta'=0.\theta'_1\ldots\theta'_n$ can be as well.
In particular, $\theta'_1\ldots\theta'_n$ is now a uniformly random $n$-bit string and $F_N\ket{\theta'_1\ldots \theta'_n}$ is a \emph{uniformly random Fourier basis state} (on which $C$ is likely to work well if its average infidelity w.r.t.\ $F_N^{-1}$ is small).

Let us first consider how to change the phase. One way to do this is to change $U$ to $U'=e^{2\pi i \lambda}U$. This has the effect that the unitary $V$ of Eq.~\eqref{eq:VcontrolledU}, with $U$ replaced by $U'$, induces an extra phase of $e^{2\pi i j \lambda}$ on basis state $\ket{j}$, resulting in 
\[
\frac{1}{\sqrt{2^n}}\sum_{j\in\01^n}e^{2\pi i j(\theta+\lambda)}\ket{j}
=\frac{1}{\sqrt{2^n}}\sum_{j\in\01^n}e^{2\pi i j \theta'}\ket{j}
=F_N\ket{\theta'_1\ldots\theta'_n}.
\]
However, a probably more efficient way to achieve the same is to leave $U$ as it is, and instead modify the $n$ Hadamard gates at the start of the phase estimation procedure. 
If we change the $\ell$th Hadamard to a single-qubit gate that maps
$$
\ket{0}\mapsto\frac{1}{\sqrt{2}}(\ket{0}+e^{2\pi i 2^{n-\ell}\lambda}\ket{1}),
$$
then the $n$-bit basis state $\ket{j}=\ket{j_1\ldots j_n}$ gets a phase 
\[
\prod_{\ell=1}^n e^{2\pi i j_\ell 2^{n-\ell}\lambda}=
e^{2\pi i (\sum_{\ell=1}^n j_\ell 2^{n-\ell})\lambda}=e^{2\pi i  j\lambda}.
\]
Thus the uniform superposition  $\frac{1}{\sqrt{2^n}}\sum_{j\in\01^n}\ket{j}$ that we prepared in the original phase estimation procedure now becomes
\[
\frac{1}{\sqrt{2^n}}\sum_{j\in\01^n}e^{2\pi i j \lambda}\ket{j}.
\]
Now we apply $V$ of Eq.~\eqref{eq:VcontrolledU} to multiply in the phases $e^{2\pi i  j\theta}$, and the $n$-qubit register becomes 
\[
\frac{1}{\sqrt{2^n}}\sum_{j\in\01^n}e^{2\pi i j (\theta+\lambda)}\ket{j}=F_N\ket{\theta'_1\ldots\theta'_n}.
\]
We have thus changed the phase from $\theta$ to $\theta'=\theta+\lambda$ as desired. Applying a perfect $F_N^{-1}$ would give us $\ket{\theta'_1\ldots\theta'_n}$ with certainty, from which we learn $\theta=\theta'-\lambda$ mod~1.

Now suppose we have a channel $C$ available that is not the perfect $F_N^{-1}$ but that has small average infidelity w.r.t.\ the perfect $F_N^{-1}$, averaged uniformly over the Fourier basis states~$\ket{\hat{k}}$:
\begin{equation}\label{eq:hifi}
\Exp_k[1-\,\bra{k}\,C(\ket{\hat{k}})\,\ket{k}\,]\leq \fidel,
\end{equation}
for some small constant $\fidel$. This could be tested by running the procedure of Theorem~\ref{th:testQFTmixed}.

Suppose we run $C$ on a specific Fourier basis state $\ket{\hat{k}}$, for some $k\in\01^n$ which would be the binary representation of the number~$\theta$. 
The intended $n$-bit outcome of a measurement in the computational basis on $C(\ket{\hat{k}})$ would be $k$. Define $\fidel_k=1-\bra{k}\,C(\ket{\hat{k}})\,\ket{k}$, which is the probability of not getting the intended outcome. 
It could be the case that we happen to run $C$ on a $\ket{\hat{k}}$ where $\fidel_k$ is particularly large; the error probability~$\fidel_k$ could even be~1 for some~$k$. In that case basic phase estimation would fail. However, we have $\Exp_{k'}[\fidel_{k'}]\leq \fidel$ by Eq.~\eqref{eq:hifi}. So if we do our worst-case-to-average-case reduction, shifting $\theta$ by a random $\lambda$ (equivalently, changing $k$ to a uniformly random $k'$ by adding a uniformly random $n$-bit integer $N\lambda$ to it, mod $N$), we obtain the following theorem:

\begin{theorem}[case where $\theta$ is an  $n$-bit number]\label{th:boundnbits}
Let $C$ be a channel from $n$ qubits to $n$ qubits with average infidelity $\leq\fidel$ w.r.t.\ $F_N^{-1}$ (in the sense of Theorem~\ref{th:testQFTmixed}). 
Let $\theta=0.\theta_1\ldots\theta_n\in\{0,1/2^n,\ldots,(2^n-1)/2^n\}$ be fixed, and $\lambda\in\{0,1/2^n,\ldots,(2^n-1)/2^n\}$ be uniformly random.
If we apply $C$ to state $\frac{1}{\sqrt{2^n}}\sum_{j\in\01^n}e^{2\pi i j(\theta+\lambda)}\ket{j}$, measure in the computational basis, and subtract $\lambda$ from the measurement outcome, then we get $\theta$ except with probability $\leq\fidel$.
\end{theorem}

As long as $\fidel<1/2$, we can reduce the error probability to an arbitrarily small $\delta$ by $O(\log(1/\delta))$ repetitions.

The above theorem is for the  basic case where  the second register contains one eigenstate $\ket{\phi}$ of~$U$, and the corresponding eigenphase can be described exactly with $n$ bits of precision.
Let us consider again the two complications mentioned near the end of Section~\ref{ssec:phasest}. The first complication is where the state $\ket{\phi}$ in the second register is a \emph{superposition} of multiple eigenstates of $U$ rather than one eigenstate, but still assuming the eigenphases can all be decribed with at most $n$ bits.
In this case we may treat the first register as containing a mixture of different states, and still use Theorem~\ref{th:boundnbits}. We will typically be using the estimated eigenphase to approximate some quantity of interest. If the obtained eigenphase would give a good approximation to that quantity with probability at least $p$ in the case of a perfect $F_N^{-1}$, then it will still give a good approximation with probability at least $p(1-\fidel)$ in the case of our imperfect channel~$C$; 
that is the situation we will be in for the application to  period-finding in Section~\ref{ssec:periodfinding}.
The second complication mentioned in Section~\ref{ssec:phasest} (when the eigenphases need more than $n$ bits) is more subtle, and we deal with it next.

\subsection{Worst-case phase estimation via average-case-correct $F_N^{-1}$: general case}\label{ssec:PEaveragegeneralcase}

If we run channel~$C$ on a \emph{superposition} of Fourier basis states in the \emph{first} register, then interference effects could occur between the different parts of the superposition when $C$ is applied, and we have to be more careful. This is the situation when we do phase estimation for the case where the phase $\theta$ needs more than $n$ bits of precision, as mentioned at the end of Section~\ref{ssec:phasest}. We will now analyse what happens in this case in more detail. 

Fortunately, by Proposition~\ref{prop:fewQFTstates} the state $\ket{\psi}$ on which we apply $C$ will still be dominated by $O(1)$ Fourier basis states, each of which corresponds to a good approximation of $\theta$. 
Choose integer $K$ and let $S$ be the set of $2K$ elements of  Proposition~\ref{prop:fewQFTstates}. 
We can write $\ket{\psi}$ as
\[
\ket{\psi}=\sum_{k\in S}\alpha_k\ket{\hat{k}}+\alpha_\rho\ket{\rho},
\]
where $\ket{\rho}$ is the normalised ``rest'' of the state (which consists of the non-$S$ Fourier basis states and hence is orthogonal to the $\ket{\hat{k}}$ with $k\in S$), and $|\alpha_\rho|^2=1-\sum_{k\in S}|\alpha_k|^2$ is small, depending on our choice of~$K$.

The phase estimation procedure ends by measuring $C(\ket{\psi})$ in the computational basis.
Every $k\in S$ will be a good measurement outcome, in the sense of corresponding to a good approximation of~$\theta$. If $C$ were perfect then $C(\ket{\psi})$ would be $\sum_k\alpha_k\ket{k}$, and measuring in the computational basis would give a good outcome~$k$ with probability $\sum_{k\in S}|\alpha_k|^2$, which is close to~1 by Proposition~\ref{prop:fewQFTstates}. However, if $C$ is not perfect then we have to worry about the interaction between the errors that $C$ makes on the different parts of the state.
The following easy lemma implies that measuring a state~$\ket{\psi}$ in a basis in which $\ket{\psi}$ has small support, cannot increase probabilities by too much.

\begin{lemma}\label{lem:measuringsmallsupport}
If $\displaystyle\ket{\psi}=\sum_{j\in T} \beta_j\ket{j}$ (not necessarily normalised), then 
$\displaystyle
\ketbra{\psi}{\psi}\preceq |T|\sum_{j\in T} |\beta_j|^2\ketbra{j}{j}.
$
\end{lemma}

\begin{proof}
Let $M$ be the matrix on the right-hand side, and consider an arbitrary state $\ket{\phi}=\sum_{j} \gamma_j\ket{j}$. Using Cauchy-Schwarz, we have
\[
\bra{\phi}(\ketbra{\psi}{\psi})\ket{\phi}=\left|\sum_{j\in T}\beta_j^*\gamma_j\right|^2
\leq |T|\sum_{j\in T}|\beta_j|^2|\gamma_j|^2=\bra{\phi}M\ket{\phi},
\]
which implies the lemma.
\end{proof}

\noindent
The prefactor $|T|$ on the right-hand side is optimal whenever all $\beta_j$'s are non-zero, as follows. Define $B=\sum_{j\in T}1/|\beta_j|^2$ and $\ket{\phi}=\frac{1}{\sqrt{B}}\sum_{j\in T}\frac{1}{\beta_j^*}\ket{j}$. Then $\bra{\phi}
\left(\sum_{j\in T} |\beta_j|^2\ketbra{j}{j}\right)\ket{\phi}=|T|/B$ while $\bra{\phi}(\ketbra{\psi}{\psi})\ket{\phi}=|\braket{\phi}{\psi}|^2=|T|^2/B$.
Hence the factor-$|T|$ on the right-hand side is necessary.

\medskip

In order to upper bound the probability of obtaining a bad outcome when measuring $C(\ketbra{\psi}{\psi})$ in the computational basis, let $P_{\rm bad}=\sum_{k\not\in S}\ketbra{k}{k}$ be the projector on the bad outcomes. Define $\fidel_k=1-\bra{k}\,C(\ket{\hat{k}})\,\ket{k}$, which is the probability of getting a measurement outcome other than $k$ when measuring $C(\ket{\hat{k}})$. We have $\Exp_k[\fidel_k]\leq \fidel$ by Eq.~\eqref{eq:hifi}.
For $k\in S$, we have $\Tr(\ketbra{k}{k}\,C(\ket{\hat{k}}))+\Tr(P_{\rm bad}\,C(\ket{\hat{k}}))\leq\Tr(C(\ket{\hat{k}})) =1$, hence $\Tr(P_{\rm bad}\,C(\ket{\hat{k}}))\leq \fidel_k$.

First consider the case that we do not shift $\theta$ by a random $\lambda$.
Below we will use Lemma~\ref{lem:measuringsmallsupport} twice, and also use the fact that the channel $C$ is linear and preserves the positive semidefinite (psd) ordering\footnote{This follows because $C$ is linear and positivity-preserving: if $\sigma\preceq\rho$, then $0\preceq C(\rho-\sigma)=C(\rho)-C(\sigma)$ and hence $C(\sigma)\preceq C(\rho)$. The latter in turn implies (in fact is equivalent to) the property that $\Tr(M\,C(\sigma))\leq \Tr(M\, C(\rho))$ for all psd operators~$M$.}, to upper bound the probability that the final measurement outcome lies outside of $S$.
\begin{align*}
    \Pr[\mbox{bad outcome}] & = \Tr\left(P_{\rm bad}\,C(\ketbra{\psi}{\psi})\right)\\
    \intertext{First apply Lemma~\ref{lem:measuringsmallsupport} with $T$ containing 2 states: $\ket{\rho}$ and the normalised version of  $\ket{\psi_S}=\sum_{k\in S}\alpha_k\ket{\hat{k}}$, so that   $\displaystyle\ket{\psi}=\norm{\ket{\psi_S}}\frac{\ket{\psi_S}}{\norm{\ket{\psi_S}}}+\alpha_\rho\ket{\rho}$. We obtain:}
     \Pr[\mbox{bad outcome}]     & \leq \Tr\left(P_{\rm bad}\,C\left(
        2\norm{\ket{\psi_S}}^2
        \frac{\ketbra
 {\psi_S}{\psi_S}}{\norm{\ket{\psi_S}}^2}+2|\alpha_\rho|^2\ketbra{\rho}{\rho}\right)
        \right)\\
   & = 2\Tr\left(P_{\rm bad}\,C\left(\ketbra{\psi_S}{\psi_S}\right)\right)\,+\,2|\alpha_\rho|^2\Tr\left(P_{\rm bad}\,C\left(\ketbra{\rho}{\rho}\right)\right)\\
       \intertext{Apply Lemma~\ref{lem:measuringsmallsupport} to $\ket{\psi_S}$ with $T=S$, and for the second term use $\Tr\left(P_{\rm bad}\,C\left(\ketbra{\rho}{\rho}\right)\right)\leq 1$:}
     & \leq 2\Tr\left(P_{\rm bad}\,C\left(|S|\sum_{k\in S} |\alpha_k|^2\ketbra{\hat{k}}{\hat{k}}\right)\right)+2|\alpha_\rho|^2\\
          & = 2|S|\left(\sum_{k\in S}|\alpha_k|^2 \, \Tr\left(P_{\rm bad}\,C(\ket{\hat{k}})\right)\right) + 2|\alpha_\rho|^2\\
     & \leq 2|S|\left(\sum_{k\in S}|\alpha_k|^2\fidel_k\right)+2|\alpha_\rho|^2.
\end{align*}

The latter upper bound on the probability of a bad outcome could be large if the $k\in S$ happen to be among the few $k$'s that have large $\fidel_k$-values. However, the average over all $N$ $\fidel_k$-values is small: at most~$\fidel$.
Now consider the case where we do use our worst-case-to-average-case reduction, shifting all $k$ by the same uniformly random $n$-bit integer $N\lambda$, and similarly shift the notion of a ``bad outcome''. Then we can upper bound the overall probability of a bad outcome by using linearity of expectation, as follows:
\begin{align}
    \Exp_{\lambda}\left[\Pr[\mbox{bad outcome}]\right] 
    & \leq \Exp_{\lambda}\left[2|S|\left(\sum_{k\in S}|\alpha_k|^2\fidel_{k+N\lambda}\right)+2|\alpha_\rho|^2
    \right]\nonumber \\ 
    & = 2|S|\left(\sum_{k\in S}|\alpha_k|^2\, \Exp_{\lambda}\left[\fidel_{k+N\lambda}\right]\right)+2|\alpha_\rho|^2\nonumber \\ 
        & \leq 2|S|\left(\sum_{k\in S}|\alpha_k|^2 \fidel\right)+2|\alpha_\rho|^2\nonumber \\  & = 2|S|(1-|\alpha_\rho|^2)\fidel + 2|\alpha_\rho|^2\nonumber \\
 & = 2|S|\fidel + 2(1-|S|\fidel)|\alpha_\rho|^2.\label{eq:errmorethannbits}
\end{align}
Note the tradeoff here:  $|S|=2K$ increases with our choice of~$K$, while $|\alpha_\rho|^2$ decreases.

We summarize the above calculations in a theorem, using 
Proposition~\ref{prop:fewQFTstates} to bound $|\alpha_\rho|^2$:

\begin{theorem}[case where $\theta$ may need more than $n$ bits of precision]\label{th:bound1}
Let $C$ be a channel from $n$ qubits to $n$ qubits with average infidelity $\leq \fidel$ w.r.t.\ $F_N^{-1}$ (in the sense of Theorem~\ref{th:testQFTmixed}). 
Let $\theta\in[0,1)$ be fixed, and $\lambda\in\{0,1/2^n,\ldots,(2^n-1)/2^n\}$ be uniformly random.
If we apply $C$ to state $\frac{1}{\sqrt{2^n}}\sum_{j\in\01^n}e^{2\pi i j(\theta+\lambda)}\ket{j}$, measure in the computational basis,
and subtract $\lambda$ from the measurement outcome, then for every integer $K\geq 1$, the probability to get an outcome $k\in\{0,\ldots,N-1\}$ such that $|\theta-k/N|>K/N$ (mod~1) is at most
\[
4K\fidel + (1/2-K\fidel)\left(\frac{1}{K}+\frac{1}{K-1}\right).
\]
\end{theorem}

For example, if we use $K=4$ (considering the $|S|=2K=8$ $k$'s that are closest to $2^n\theta$ to be the good measurement outcomes) and we set $\fidel\leq 0.015$, then the probability of a bad outcome will be $<0.497$. Thus the probability of obtaining a good approximation of the phase~$\theta$ is $>0.503$, and we can amplify this success probability to be close to~1 by taking the median of multiple runs of this procedure. 

We can improve these bounds on the tolerable average  infidelity~$\fidel$ a bit by plugging the numerical upper bounds on $|\alpha_\rho|^2$ from the end of the appendix into Eq.~(\ref{eq:errmorethannbits}), taking $N=2^{10}$ for concreteness (for other values of $N$ the numbers are very similar).
With $K=4$ and using $|\alpha_\rho|^2\approx 0.05$ from the appendix, the probability of a bad outcome is $\leq 16\fidel + 2(1-8\fidel)0.05$, 
which is $<0.5$ as long as the average infidelity $\fidel$ is $\leq 0.026$.
With $K=3$, using $|\alpha_\rho|^2\approx 0.067$ from the appendix, the probability of a bad outcome is $\leq 12\fidel + 2(1-6\fidel)0.067$, 
which is $<0.5$ as long as the average infidelity $\fidel$ is $\leq 0.032$.
With $K=2$, using  $|\alpha_\rho|^2\approx 0.099$ from the appendix, we can tolerate $\fidel\leq 0.041$.
Even with these numerical improvements, the tolerable $\fidel$ is still much smaller here than in the case where $\theta$ is an $n$-bit number (Theorem~\ref{th:boundnbits}): there we could tolerate any $\fidel<1/2$.

\section{Applications: period-finding and amplitude estimation}\label{sec:applications}

\subsection{Period-finding}\label{ssec:periodfinding}

Fix $N=2^n$.\footnote{To avoid confusion with Shor's algorithm: our $N$ here is the dimension of the QFT, not an integer to be factored.} For fixed \emph{period} $r<N$ and variable \emph{offset} $s\in\{0,\ldots,r-1\}$, define periodic state
\[
\ket{\pi_{s}}=\frac{1}{\sqrt{p}}\sum_{z=0}^{p-1}\ket{s+zr},
\]
where $p=|\{z: 0\leq s+zr<N\}|$ is the number of basis states occurring in the superposition. This $p$ will be $N/r$ rounded up; 
$N$ has $\floor{N/r}$ ``complete'' sequences of $r$ indices followed by one ``incomplete'' sequence of $N-r\floor{N/r}$ indices.
Shor~\cite{shor:factoring} showed via classical number theory that the ability to find the period~$r$ given such a state suffices for factoring integers, and gave an efficient quantum algorithm for period-finding. Subsequently Kitaev~\cite{kitaev:stabilizer} showed how to do period-finding using phase estimation, and then Cleve, Ekert, Mosca, and Macchiavello~\cite{cemm:revisited} showed that Shor's and Kitaev's approaches to period-finding are basically the same.

In order to do period-finding via phase estimation starting from a periodic state, we let $U$ be the ``$+1 \bmod N$'' operator:
\[
U\ket{x}=\ket{x+1\bmod N}.
\]
It is easily verified that the eigenstates of $U$ are the states $F_N^{-1}\ket{j}$ with corresponding eigenvalue~$\omega_N^j=e^{2\pi i j/N}$. Note that $n$ bits of precision suffice for each of these eigenphases.

Because these eigenstates form a basis, any state $\ket{\phi}$ can be written as a superposition $\sum_{j=0}^{N-1} \alpha_j F_N^{-1}\ket{j}$ of eigenstates of~$U$ with some coefficients~$\alpha_j$. We already saw how phase estimation using $U$ acts when we start with such a superposition in the second register: with probability $|\alpha_j|^2$ it returns the $n$-bit number~$j/N$ exactly.
We now determine these $\alpha_j$ coefficients for the case where our starting state is the periodic state $\ket{\pi_s}=\sum_j \alpha_j F_N^{-1}\ket{j}$, by multiplying $\ket{\pi_s}$ with $F_N$ (we do this only as a calculational device; we are not implementing the QFT physically):
\begin{align*}
\sum_{j=0}^{N-1}\alpha_j \ket{j} & = F_N\ket{\pi_s}
  = \frac{1}{\sqrt{p}}\sum_{z=0}^{p-1}F_N\ket{s+zr}\\
  & = \frac{1}{\sqrt{p}}\sum_{z=0}^{p-1}\frac{1}{\sqrt{N}}\sum_{j=0}^{N-1}\omega_N^{j(s+zr)}\ket{j}
  = \sum_{j=0}^{N-1}\underbrace{\frac{\omega_N^{js}}{\sqrt{pN}}\sum_{z=0}^{p-1}\omega_N^{jzr}}_{\alpha_j}\ket{j}.
\end{align*}
We now want to show that the amplitude is concentrated around integer multiples of $N/r$.
First consider the special case where $N/r$ happens to be an integer. If $j=cN/r$ for some integer $c\in\{0,\ldots,r-1\}$, then  we have $\omega_N^{jzr}=1$ for all $z$ and hence $|\alpha_j|^2=p/N=1/r$; there are $r$ such $j$'s, each with squared amplitude $1/r$, so the $j$'s that are not integer multiples of $N/r$ will have amplitude~0 in this case.

In the general case where $N/r$ is not an integer, let $j$ be the closest integer to $cN/r$ for some $c\in\{0,\ldots,r-1\}$ (i.e., $j=cN/r+\delta$ for $\delta\in(-1/2,1/2]$). Then 
\begin{align*}
|\alpha_j|^2
& =\frac{1}{pN}\left|\sum_{z=0}^{p-1}\omega_N^{jzr}\right|^2
=\frac{1}{pN}\frac{|1-\omega_N^{pjr}|^2}{|1-\omega_N^{jr}|^2}
=\frac{1}{pN}\frac{|1-e^{2\pi i p\delta r/N}|^2}{|1-e^{2\pi i \delta r/N}|^2}
=\frac{1}{pN}\frac{\sin(\pi p \delta r/N)^2}{\sin(\pi \delta r/N)^2}\\
& \geq\frac{1}{pN}\frac{(\frac{2}{\pi}\pi p \delta r/N)^2}{(\pi \delta r/N)^2}=\frac{4p}{\pi^2 N}\geq\frac{4}{\pi^2 r},
\end{align*}
using $\frac{2}{\pi}x\leq\sin(x)\leq x$ for $x\in[0,\pi/2]$, and assuming $\delta\neq 0$. If indeed $\delta\neq 0$, then the probability that the measurement outcome~$j$ is one of the two integers in the interval $(cN/r-1,cN/r+1)$ is $\geq 8/(\pi^2 r)$.\footnote{It need not be the case that each of the two outcomes $\ceil{cN/r}$ and $\floor{cN/r}$ has probability $\geq 4/(\pi^2 r)$, because one of them will correspond to a $\delta\not\in(-1/2,1/2]$. However, the sum of these two probabilities is at least $8/(\pi^2 r)$, which may be verified by noting that the function $f(x)=\sin(\pi x)^2/\sin(\pi x/p)^2+\sin(\pi (1-x))^2/\sin(\pi(1-x)/p)^2$ is at least $8p^2/\pi^2$ on the interval $x\in[0,1/2]$. We note that this calculation is similar to that in the appendix; here, however, we have a superposition of eigenstates each with $n$-bit eigenphase; in the appendix similar calculations are needed to deal with the situation that the eigenphase needs more than $n$ bits of precision.} If $\delta=0$, then $|\alpha_j|^2=p/N\geq 1/r$.
Accordingly, with probability $\geq 8/\pi^2$ our measurement outcome~$j$ is $cN/r$ (rounded up or down) for some random integer $c\in\{0,\ldots,r-1\}$.
Note that in that case we have $|j/N - c/r|< 1/N$, so the known ratio $j/N$ is a very good approximation to the unknown ratio $c/r$.
Shor showed that if $c$ and $r$ are coprime (which, by classical number theory, happens with largish probability $\Omega(1/\log\log r)$), then continued-fraction expansion on the known ratio $j/N$ yields the period~$r$ (as mentioned, this suffices for factoring).

Using our worst-case to average-case reduction, this method for period-finding still works when we only have a good-on-average inverse QFT for our phase estimation, provided the initial periodic state $\ket{\pi_s}$ is prepared sufficiently well and the other components of phase estimation (the unitary $V$ from Eq.~(\ref{eq:VcontrolledU}) and the $n$ modified Hadamard gates) also work sufficiently well.
Note that the relevant eigenphases can all be described by $n$ bits here, so we only need to refer to the result in Section~\ref{ssec:PEaveragenbitcase} and not to the more complicated result in Section~\ref{ssec:PEaveragegeneralcase}, where the tolerable $\fidel$ is worse.
By the result in Section~\ref{ssec:PEaveragenbitcase}, since the probability to find a good $j$ is $8/\pi^2$ using a perfect inverse QFT, this probability is still $\geq (1-\fidel)8/\pi^2$ if we instead use a channel~$C$ that has average infidelity $\leq\fidel$ w.r.t.\ the perfect inverse QFT.

\subsection{Amplitude estimation}

Suppose we have a unitary quantum algorithm $A$ that maps
\[
\ket{0^m}\mapsto A\ket{0^m}= \sin(\mu)\ket{\phi_1}\ket{1}+
\cos(\mu)\ket{\phi_0}\ket{0},
\]
for some arbitrary normalised states $\ket{\phi_1}$ and $\ket{\phi_0}$, and some angle $\mu\in[0,\pi/2]$. We would like to estimate the angle $\mu$ or (which comes to the same thing) the amplitude $\sin(\mu)$. 
Such ``amplitude estimation'' can be used for instance for optimal quantum approximate counting~\cite{bhmt:countingj} and is a key component of many other quantum algorithms that involve estimating various quantities. 
For completeness we here sketch the elegant method of Brassard et al.~\cite{bhmt:countingj} that reduces amplitude estimation to phase estimation. Consider the unitary
\[
U=AR_0A^{-1}(I\otimes Z),
\]
where $R_0=2\ketbra{0^m}{0^m}-I$ reflects about $\ket{0^m}$. This $R_0$ can be implemented using a circuit with $O(m)$ elementary gates.
Consider how $U$ acts in the 2-dimensional space $\cal S$ spanned by 
$\ket{\phi_1}\ket{1}$ and
$\ket{\phi_0}\ket{0}$.\footnote{It is very helpful to picture this similarly to the usual analysis of Grover's algorithm, in a 2d plane where the vertical axis corresponds to the state $\ket{\phi_1}\ket{1}$ and the horizontal axis corresponds to
$\ket{\phi_0}\ket{0}$.}
$U$ is the product of two reflections:
first $I\otimes Z$ reflects about $\ket{\phi_0}\ket{0}$, and then 
$AR_0A^{-1}$ reflects about $A\ket{0^m}$. This product of two reflections corresponds (in the space $\cal S$) to a rotation over twice the angle $\mu$ that exists between $\ket{\phi_0}\ket{0}$ and $A\ket{0^m}= \sin(\mu)\ket{\phi_1}\ket{1}+
\cos(\mu)\ket{\phi_0}\ket{0}$. Such a rotation over angle $2\mu$ has two eigenvectors in~$\cal S$, with respective eigenvalues $e^{i2\mu}$ and $e^{-i2\mu}$.

How do we use phase estimation to estimate $\mu$?
We start in state $A\ket{0^m}$, which lies in $\cal S$ and hence is some linear combination of the two eigenvectors (with unknown coefficients, but that doesn't matter).
If we run phase estimation, then the output will either be an estimate of $\mu/\pi$ or of  $-\mu/\pi$ (or rather, $1-\mu/\pi$). Since we assumed $\mu\in[0,\pi/2]$, the phase $\mu/\pi$ that we're estimating lies in $[0,1/2]$.
If our estimate is in $[0,1/2]$ then we'll assume it's $\mu/\pi$, and if our estimate is in $[1/2,1)$ then we'll assume it's $1-\mu/\pi$. 
Either way we obtain a good estimate of $\mu$. 

This method still works with an only-good-on-average channel for $F_N^{-1}$ if we do our worst-case to average-case reduction.
If we want $n$ bits of precision in our estimate of $\mu$, then the cost will be $O(2^n)$ applications of~$U$. Note that here we need the more complicated result of Section~\ref{ssec:PEaveragegeneralcase}, since  the eigenphases $\pm\mu/\pi$ may require more than $n$ bits of precision.

\section{Summary and future work}

In this paper we did two things.
First, we showed that one can efficiently test whether a given $n$-qubit channel~$C$ implements the inverse QFT well, on average over all Fourier basis states. Our procedure estimates (with success probability $\geq 1-\delta$) the average (in)fidelity up to $\pm\eps$ using $O(\log(1/\delta)/\eps^2)$ runs, each of which uses $O(n)$ single-qubit gates to prepare a product state, one run of $C$, and a measurement in the computational basis.

Second, we showed that such an average-case-correct inverse QFT suffices to implement phase estimation \emph{in the worst case}.
This implies that an average-case-correct inverse QFT also suffices for period-finding (as in Shor's algorithm) and for amplitude estimation, provided the other components of those procedures work sufficiently well.

Practical methods of verification for large numbers of qubits will be vital for future quantum computers. Here we have shown that such verification is possible for several key algorithmic primitives.
We feel it would be very interesting to find more examples of worst-case-to-average-case reductions for quantum computing. 

One more example would be the quantum algorithm of Jordan~\cite{jordan:gradient} for computing the gradient of a differentiable function $f:\R^d\to\R$, i.e., the vector of the $d$ partial derivatives of $f$ evaluated at a given point. Let's say we want to approximate the gradient with $n$ bits of precision for each of its $d$ entries. Jordan's algorithm sets up a uniform superposition over some grid of inputs $(x_1,\ldots,x_d)$ in $d$ $n$-qubit registers, then computes $f$ once in the phase, and then applies an inverse QFT to each of the $d$ registers. If $f$ is affine-linear, i.e., there are coefficients $a,b_1,\ldots,b_d\in\R$ such that  $f(x)=a+b_1x_1+\cdots+b_dx_d$ for all $x\in\R^n$, and each $b_i$ can be represented exactly with $n$ bits of precision, then the final state gives $b_1,\ldots,b_d$ which is exactly the gradient of $f$. If some of the $b_i$'s need more than $n$ bits of precision then at the end of the algorithm the $i$th register contains $b_i$ with high probability. This algorithm ``costs'' only one $f$-evaluation, $d$ runs of the $n$-qubit inverse QFT, and $O(dn)$ other elementary gates.
In contrast, classical algorithms need $d$ $f$-evaluations to compute the gradient.
In the more general case where $f$ is fairly smooth but only close to an affine-linear function on the grid of points that we feed into it, a much more complicated analysis due to Gily\'{e}n, Arunachalam, and Wiebe~\cite{GAW:gradient} determines the complexity of approximating the gradient in various norms, in terms of how many times $f$ needs to be evaluated.

What if we only have a good-on-average inverse QFT available?
If $f$ is affine-linear then one may confirm that the analysis of our Section~\ref{ssec:PEaveragegeneralcase} implies that Jordan's algorithm can still be made to work via our worst-case-to-average-case reduction.
If $f$ is only close to affine-linear then it is not clear what happens, and whether the more subtle analysis of~\cite{GAW:gradient} can also be made to work with an imperfect inverse QFT. We leave this question to further work.

\subsection*{Acknowledgements.}
We thank Timothy Browning for helpful discussions regarding some number theory for Shor's algorithm, and the anonymous Quantum referees for very helpful comments that substantially improved the presentation and rigour of the paper.

NL was partially supported by the UK Engineering and Physical Sciences Research Council through grants EP/R043957/1, EP/S005021/1, EP/T001062/1.
RdW was partially supported by the Dutch Research Council (NWO) through Gravitation-grant Quantum Software Consortium, 024.003.037, and through QuantERA ERA-NET Cofund project QuantAlgo 680-91-034.

This paper did not involve any underlying data.

\bibliographystyle{alphaUrlePrint}
\bibliography{qc}

\newcommand{\etalchar}[1]{$^{#1}$}
\begin{thebibliography}{EHW{\etalchar{+}}20}

\bibitem[AR20]{aaronson&rall:qcounting}
Scott Aaronson and Patrick Rall.
\newblock \href{http://dx.doi.org/10.1137/1.9781611976014.5}{Quantum
  approximate counting, simplified}.
\newblock In {\em Proceedings of 3rd Symposium on Simplicity in Algorithms
  (SOSA)}, pages 24--32, 2020.
\newblock arXiv:1908.10846.

\bibitem[BBBV97]{bbbv:str&weak}
Charles~H. Bennett, Ethan Bernstein, Gilles Brassard, and Umesh Vazirani.
\newblock \href{http://dx.doi.org/doi.org/10.1137/S0097539796300933}{Strengths
  and weaknesses of quantum computing}.
\newblock {\em SIAM Journal on Computing}, 26(5):1510--1523, 1997.
\newblock quant-ph/9701001.

\bibitem[BHMT02]{bhmt:countingj}
Gilles Brassard, Peter H{\o}yer, Michele Mosca, and Alain Tapp.
\newblock \href{http://dx.doi.org/doi.org/10.1090/conm/305/05215}{Quantum
  amplitude amplification and estimation}.
\newblock In {\em Quantum Computation and Quantum Information: A Millennium
  Volume}, volume 305 of {\em AMS Contemporary Mathematics Series}, pages
  53--74. 2002.
\newblock quant-ph/0005055.

\bibitem[CB21]{chen&brandao:trotter}
Chi-Fang Chen and Fernando G. S.~L. Brand{\~{a}}o.
\newblock
  \href{http://dx.doi.org/doi.org/10.48550/arXiv.2111.05324}{Concentration for
  {T}rotter error}.
\newblock arXiv:2111.05324, 9 Nov 2021.

\bibitem[CEMM98]{cemm:revisited}
Richard Cleve, Artur Ekert, Chiara Macchiavello, and Michele Mosca.
\newblock \href{http://dx.doi.org/doi.org/10.1098/rspa.1998.0164}{Quantum
  algorithms revisited}.
\newblock In {\em Proceedings of the Royal Society of London}, volume A454,
  pages 339--354, 1998.
\newblock quant-ph/9708016.

\bibitem[Cop94]{coppersmith:fourier}
Don Coppersmith.
\newblock \href{http://dx.doi.org/doi.org/10.48550/arXiv.quant-ph/0201067}{An
  approximate {F}ourier transform useful in quantum factoring}.
\newblock IBM Research Report No.~RC19642, quant-ph/0201067, 1994.

\bibitem[dSLCP11]{SLP:practical}
Marcus da~Silva, Oliver Landon-Cardinal, and David Poulin.
\newblock \href{http://dx.doi.org/10.1103/PhysRevLett.107.210404}{Practical
  characterization of quantum devices without tomography}.
\newblock {\em Physical Review Letters}, 107:210404, 2011.
\newblock arXiv:1104.3835.

\bibitem[EHW{\etalchar{+}}20]{eisertetal:certificationsurvey}
Jens Eisert, Dominik Hangleiter, Nathan Walk, Ingo Roth, Damian Markham, Rhea
  Parekh, Ulysse Chabaud, and Elham Kashefi.
\newblock \href{http://dx.doi.org/10.1038/s42254-020-0186-4}{Quantum
  certification and benchmarking}.
\newblock {\em Nature Reviews Physics}, 2:382--390, 2020.
\newblock arXiv:1910.06343.

\bibitem[FL11]{flammia&liu:fidelityestimation}
Steven~T. Flammia and Yi-Kai Liu.
\newblock \href{http://dx.doi.org/10.1103/PhysRevLett.106.230501}{Direct
  fidelity estimation from few {P}auli measurements}.
\newblock {\em Physical Review Letters}, 106:230501, 2011.
\newblock arXiv:1104.4695.

\bibitem[GAW19]{GAW:gradient}
Andr\'{a}s Gily{\'e}n, Srinivasan Arunachalam, and Nathan Wiebe.
\newblock
  \href{http://dx.doi.org/doi.org/10.1137/1.9781611975482.87}{Optimizing
  quantum optimization algorithms via faster quantum gradient computation}.
\newblock In {\em Proceedings of 30th ACM-SIAM SODA}, pages 1425--1444, 2019.
\newblock arXiv:1711.00465.

\bibitem[Gro96]{grover:search}
Lov~K. Grover.
\newblock \href{http://dx.doi.org/doi.org/10.1145/237814.237866}{A fast quantum
  mechanical algorithm for database search}.
\newblock In {\em Proceedings of 28th ACM STOC}, pages 212--219, 1996.
\newblock quant-ph/9605043.

\bibitem[GSLW19]{gilyenea:svtrans}
Andr{\'a}s Gily{\'e}n, Yuan Su, {Guang Hao} Low, and Nathan Wiebe.
\newblock \href{http://dx.doi.org/10.1145/3313276.3316366}{Quantum singular
  value transformation and beyond: exponential improvements for quantum matrix
  arithmetics}.
\newblock In {\em Proceedings of 51st ACM STOC}, pages 193--204, 2019.
\newblock arXiv:1806.01838.

\bibitem[Jor05]{jordan:gradient}
Stephen~P. Jordan.
\newblock \href{http://dx.doi.org/doi.org/10.1103/PhysRevLett.95.050501}{Fast
  quantum algorithm for numerical gradient estimation}.
\newblock {\em Physical Review Letters}, 95:050501, 2005.
\newblock quant-ph/0405146.

\bibitem[Kit95]{kitaev:stabilizer}
Alexey~Yu. Kitaev.
\newblock
  \href{http://dx.doi.org/doi.org/10.48550/arXiv.quant-ph/9511026}{Quantum
  measurements and the {A}belian stabilizer problem}.
\newblock quant-ph/9511026, 12 Nov 1995.

\bibitem[LW21]{linden&wolf:lightweight}
Noah Linden and {Ronald de} Wolf.
\newblock
  \href{http://dx.doi.org/doi.org/10.22331/q-2021-04-20-436}{Lightweight
  detection of a small number of large errors in a quantum circuit}.
\newblock {\em Quantum}, 5(436), 2021.
\newblock arXiv:2009.08840.

\bibitem[Mah18]{mahadev:clasveri}
Urmila Mahadev.
\newblock \href{http://dx.doi.org/10.1109/FOCS.2018.00033}{Classical
  verification of quantum computations}.
\newblock In {\em Proceedings of 59th IEEE FOCS}, pages 259--267, 2018.
\newblock arXiv:1804.01082.

\bibitem[MRTC21]{MRTC:grandunif}
John~M. Martyn, Zane~M. Rossi, Andrew~K. Tan, and Isaac~L. Chuang.
\newblock \href{http://dx.doi.org/doi.org/10.1103/PRXQuantum.2.040203}{A grand
  unification of quantum algorithms}.
\newblock {\em PRX Quantum}, 2:040203, 2021.
\newblock arXiv.2105.02859.

\bibitem[NC00]{nielsen&chuang:qc}
Michael~A. Nielsen and Isaac~L. Chuang.
\newblock \href{http://dx.doi.org/doi.org/10.1017/CBO9780511976667}{{\em
  Quantum Computation and Quantum Information}}.
\newblock Cambridge University Press, 2000.

\bibitem[Ral21]{rall:phaseest}
Patrick Rall.
\newblock \href{http://dx.doi.org/doi.org/10.22331/q-2021-10-19-566}{Faster
  coherent quantum algorithms for phase, energy, and amplitude estimation}.
\newblock {\em Quantum}, 5(566), 2021.
\newblock arXiv:2103.09717.

\bibitem[Sho97]{shor:factoring}
Peter~W. Shor.
\newblock
  \href{http://dx.doi.org/doi.org/10.1137/S0097539795293172}{Polynomial-time
  algorithms for prime factorization and discrete logarithms on a quantum
  computer}.
\newblock {\em SIAM Journal on Computing}, 26(5):1484--1509, 1997.
\newblock Earlier version in FOCS'94. quant-ph/9508027.

\bibitem[ZZS{\etalchar{+}}21]{Zhaoetal:HamSimRandom}
Qi~Zhao, You Zhou, Alexander~F. Shaw, Tongyang Li, and Andrew~M. Childs.
\newblock
  \href{http://dx.doi.org/doi.org/10.48550/arXiv.2111.04773}{Hamiltonian
  simulation with random inputs}.
\newblock arXiv:2111.04773, 8 Nov 2021.

\end{thebibliography}

\appendix

\section{Proof of Proposition \ref{prop:fewQFTstates}}

Let $N=2^n$, $\theta\in[0,1)$, and let coefficients $\alpha_k\in\mathbb{C}$ be such that 
\begin{equation}
\frac{1}{\sqrt{N}}\sum_{j\in\01^n}e^{2\pi i j\theta}\ket{j}
=\sum_{k=0}^{N-1}\alpha_k\ket{\hat{k}}. \label{eq:fourier1}
\end{equation}
Let $k^*=\floor{N\theta}\in\{0,\ldots,N-1\}$ be the $n$-bit integer corresponding to the first $n$ bits in the binary expansion of $\theta$, $x=N\theta-k^*\in[0,1)$ be the fractional part of $N\theta$. Let $S=\{k^*-K+1,\ldots,k^*-1,k^*,k^*+1,\ldots,k^*+K\}$ be the set of $2K$ integers ``around'' $2^n\theta$ (these integers should be taken mod~$N$). 
These $2K$ integers come in $K$ pairs: $\{k^*,k^*+1\}, \{k^*-1,k^*+2\}, \{k^*-2,k^*+3\},\ldots, \{k^*-K+1,k^*+K\}$.

Our goal is to show that the total probability $P=\sum_{k\in S}|\alpha_k|^2$ of terms in the set $S$ is close to~1. By applying inverse QFT to both sides of Eq.~\eqref{eq:fourier1} and rewriting the geometric series, we have
\[ 
\alpha_k= \frac{1}{N} \frac{e^{2\pi i (N\theta-k)}-1}{e^{2\pi i (N\theta-k)/N}-1} \mbox{, hence }
| \alpha_k|^2= \frac{1}{N^2} \frac{\sin^2(\pi (N\theta-k))}{\sin^2(\pi  (N\theta-k)/N)}=\frac{1}{N^2} \frac{\sin^2(\pi x)}{\sin^2(\pi  (N\theta-k)/N)}.
\]
So
\begin{align} \label{eq:Pseries}
P= \frac{1}{N^2} \sum_{k=k^*-K+1}^{k^*+K}
\frac{\sin^2(\pi x)}{\sin^2(\pi  (k^*-k+x)/N)}
= 
\frac{1}{N^2}\sum_{k=-K+1}^{K}
\frac{\sin^2(\pi x)}{\sin^2(\pi  (x-k)/N)}
\end{align}

\bigskip

We will now prove an upper bound on the probability of error (the analysis is similar to that in \cite{nielsen&chuang:qc} section 5.2.1). 
\begin{align} 
1-P & =
\frac{1}{N^2}\sum_{k=-N/2+1}^{-K}
\frac{\sin^2(\pi x)}{\sin^2(\pi  (x-k)/N)}
+
\frac{1}{N^2}\sum_{k=K+1}^{N/2}
\frac{\sin^2(\pi x)}{\sin^2(\pi  (x-k)/N)}
\nonumber \\
&\leq 
\frac{\sin^2(\pi x)}{4}\Big(\sum_{k=-N/2+1}^{-K}
\frac{1}{(x-k)^2} +
\sum_{k=K+1}^{N/2}
\frac{1}{(x-k)^2}
\Big)\nonumber\\
&=
\frac{\sin^2(\pi x)}{4} \sum_{k=K+1}^{N/2}
\Big(\frac{1}{(k-x)^2}  + \frac{1}{(k-1+x)^2} 
\Big)\nonumber\\
&\leq
\frac{\sin^2(\pi x)}{4}
 \int_{K}^{N/2}dt
\Big(\frac{1}{(t-x)^2}  + \frac{1}{(t-1+x)^2}\Big)
\nonumber\\
&\leq
\frac{\sin^2(\pi x)}{4}
 \int_{K}^{\infty}dt
\Big(\frac{1}{(t-x)^2}  + \frac{1}{(t-1+x)^2}\Big)
\nonumber\\
&=
\frac{\sin^2(\pi x)}{4}
\Big(\frac{1}{(K-x)}  + \frac{1}{(K-1+x)}\Big)
\nonumber\\
&\leq
\frac{1}{4}
\Big(\frac{1}{(K-x)}  + \frac{1}{(K-1+x)}\Big)
\nonumber\\
&\leq
\frac{1}{4}
\Big(\frac{1}{K}  + \frac{1}{K-1}\Big).
\label{eq:provable-bound1}
\end{align}

The final inequality is easy to see because the function is maximised on $[0,1]$ at the two endpoints of that interval ($x=0$ or $x=1$).
We use this bound in Proposition~\ref{prop:fewQFTstates} in the body of the paper.

\bigskip

In fact, numerical evidence suggests that the above rigorous bound is rather loose. The probability of Eq.~(\ref{eq:Pseries}) is symmetric around $x=1/2$, because replacing $x$ by $1-x$ leaves $P$ invariant (replacing $x$ by $1-x$ interchanges the two squared amplitudes $|\alpha_{k^*+k}|^2$ and $|\alpha_{k^*-k+1}|^2$, leaving the sum of the pair invariant).
Hence $P$ has a local optimum at $x=1/2$. Plotting the graphs for small $K$ strongly suggests that  $P$ is actually minimised at $x=1/2$.  If this is indeed the case, then (using the symbol $\tilde P$ to denote this conjectured bound)
\begin{align} 
\tilde P&\geq
\frac{1}{N^2}\sum_{k=-K+1}^{K}
\frac{1}{\sin^2(\pi  (\frac{1}{2}-k)/N)}\nonumber \\
&=
\frac{1}{N^2}\sum_{k=-K}^{K-1}
\frac{1}{\sin^2(\pi  (\frac{1}{2}+k)/N)}\nonumber \\
&\geq \frac{4}{\pi^2}\sum_{k=-K}^{K-1}
\frac{1}{(2k+1)^2}\nonumber \\
&= \frac{8}{\pi^2}\sum_{k=0}^{K-1}
\frac{1}{(2k+1)^2}
\nonumber \\
&= \frac{8}{\pi^2}\Big(\sum_{k=0}^{\infty}
\frac{1}{(2k+1)^2}
-\sum_{k=K}^{\infty}  
\frac{1}{(2k+1)^2}\Big)\nonumber\\
&= \frac{8}{\pi^2}\Big(\frac{\pi^2}{8}
-\sum_{k=K}^{\infty}  
\frac{1}{(2k+1)^2}\Big)
\nonumber\\
&\geq \frac{8}{\pi^2}\Big(\frac{\pi^2}{8}
-\int_{K-1}^\infty \frac{dy}{(2y+1)^2}\Big)
\nonumber\\
&=1-\frac{4}{\pi^2(2K-1)}
\label{eq:PLB}
\end{align}
Numerical evidence shows that this  asymptotic bound is reasonably close to the sum for moderate-sized and large~$K$, but still somewhat off for small~$K$:
\begin{itemize}
\item For $K=2$, the proven upper bound for the error of Eq.~(\ref{eq:provable-bound1}) is 0.375, whereas the upper bound for the error using (\ref{eq:PLB}) is 0.135, and the numerical evaluation of the exact formula (\ref{eq:Pseries}) for $1-P$ for $K=2$, e.g.\ $N=2^{10}$ and $x=1/2$, is 0.099.
\item For $K=3$, the proven upper bound for the error of Eq.~(\ref{eq:provable-bound1}) is 0.208, whereas the upper bound for the error using (\ref{eq:PLB}) is 0.081, and the numerical evaluation of the exact formula (\ref{eq:Pseries}) for $1-P$ for $K=3$, e.g.\ $N=2^{10}$ and $x=1/2$, is 0.067.
\item For $K=4$, the proven upper bound for the error of Eq.~(\ref{eq:provable-bound1}) is 0.146, whereas the upper bound for the error using (\ref{eq:PLB}) is 0.058, and the numerical evaluation of the exact formula (\ref{eq:Pseries}) for $1-P$ for $K=4$, e.g.\ $N=2^{10}$ and $x=1/2$, is 0.050.
\end{itemize}

\end{document}